\documentclass[10pt]{article}
\usepackage{}

\usepackage{pgf}
\usepackage[square, comma, sort&compress, numbers]{natbib}
\usepackage[english]{babel}
\usepackage[utf8]{inputenc}
\usepackage{booktabs}
\usepackage{caption}
\usepackage[T1]{fontenc}
\usepackage[font=small,labelfont=bf,tableposition=top]{caption}
\usepackage{threeparttable}

\usepackage{datetime}
\usepackage{geometry}
\usepackage{array}
\usepackage{xspace}

\usepackage{amsmath,amsfonts,amssymb,amsopn,amsthm,amscd}
\theoremstyle{plain}
\usepackage{graphicx} 
\usepackage{epstopdf}
\usepackage{enumerate}

\def\Box{\vcenter{\vbox{\hrule\hbox{\vrule
     \vbox to 8.8pt{\hbox to 10pt{}\vfill}\vrule}\hrule}}}

\usepackage{indentfirst}
\newcommand{\Ff}{{\mathbb F}}

\newcommand{\Zz}{{\mathbb Z}}
\newcommand{\Cc}{{\mathbb C}}

\newcommand\cc{{\mathcal C}}        %

\newtheorem{thm}{Theorem}[section]
\newtheorem{lem}[thm]{Lemma}

\newtheorem{remark}{Remark}

\begin{document}

%

\title{MDS codes with arbitrary dimensional hull and their applications
}

\author{Gaojun Luo$^1$, Xiwang Cao$^{1,2}$\thanks{This work was supported by the National Natural Science Foundation of China (Grant No. 11771007 and 61572027).}}

\maketitle

\let\thefootnote\relax\footnotetext{$^1$Department of Math, Nanjing University of Aeronautics and Astronautics, Nanjing 211100, China. \\(Email: gjluo1990@163.com)}

\let\thefootnote\relax\footnotetext{$^2$State Key Laboratory of Information Security, Institute of Information Engineering, Chinese Academy of Sciences, Beijing 100093, China. (Email: xwcao@nuaa.edu.cn)}

\begin{abstract}
The hull of linear codes have promising utilization in coding theory and quantum coding theory. In this paper, we study the hull of generalized Reed-Solomon codes and extended generalized Reed-Solomon codes over finite fields with respect to the Euclidean inner product. Several infinite families of MDS codes with arbitrary dimensional hull are presented. As an application, using these MDS codes with arbitrary dimensional hull, we construct several new infinite families of entanglement-assisted quantum error-correcting codes with flexible parameters.

{\bf Keywods}: Hull, generalized Reed-Solomon code, MDS code, Entanglement-assisted quantum error-correcting code (EAQECC)

\end{abstract}

\section{Introduction}
Let $q$ be a power of a prime and $\Ff_q$ denote the finite field with $q$ elements. An $[n,k,d]$ linear code over $\Ff_q$ is a $k$-dimensional subspace of $\Ff_q^n$ with minimum Hamming distance $d$. Let $\Ff_q^n$ stand for the vector space with dimension $n$ over $\Ff_q$. Maximum distance separable (MDS) codes are optimal in the sense that no code of length $n$ with $K$ codewords has a larger minimum distance than that of a MDS code with length $n$ and size $K$. Mathematically, an $[n,k,d]$ code $\cc$ is called a MDS code if $n=k+d-1$. For any two vectors $\mathbf{x}=(x_1,x_2,\cdots,x_n)^t$ and $\mathbf{y}=(y_1,y_2,\cdots,y_n)^t$ of $\Ff_q^n$, their Euclidean inner product is defined as
$$
\mathbf{x}\cdot\mathbf{y}=\sum_{i=1}^nx_iy_i.
$$
The dual of the code $\cc$ is defined by the set
$$
\cc^{\bot}=\{\mathbf{x}\in\Ff_q^n:\mathbf{x}\cdot\mathbf{y}=0\ {\rm for\ all }\ \ \mathbf{y}\in\cc\}.
$$
The hull of $\cc$ is the code $\cc\cap \cc^{\bot}$, denoted by Hull$(\cc)$, in the terminology that was introduced in \cite{Hull}. If Hull$(\cc)=\{0\}$, then the linear code $\cc$ is termed a linear complementary dual (LCD) code. Recently, the study of LCD codes has attracted much attention due to their applications in orthogonal direct sum masking (ODSM), protecting against side-channel attacks (SCAs) and fault injection attacks (FIAs). The existence question about MDS codes with complementary duals over $\Ff_q$ has been completely addressed in \cite{Carlet} and \cite{Jin} for $q=2$ and $q>3$, respectively.

The test of the permutation equivalence of two codes and the determination of the automorphism group of a linear code are interesting problems in coding theory \cite{A1,A2}. Some algorithms for these computations have been provided in \cite{A3,A4,A5,A6}. The complexity of these algorithms is determined by the dimension of the hull of codes. Consequently, the study of the dimension and properties of hull of codes is useful for these computations. In \cite{R1}, Sendrier established the number of distinct $q$-ary linear codes of length $n$ with a given dimensional hull. Skersys \cite{R2} discussed the average dimension of the hull of cyclic codes. Recently, Sangwisut et al. \cite{R3} have studied the hull of cyclic and negacyclic codes over finite fields.

In quantum information, the existence of quantum error correcting codes (QECCs) was one of the most important discoveries in 1995. Afterwards, Calderbank, Shor and Steane \cite{CSS1,CSS2} provided a method for constructing QECCs (namely, the CSS construction), which establishes the connections between quantum stabilizer codes and classical linear codes. In the CSS construction for quantum codes, the classical linear codes need to be dual-containing, otherwise the resulting ``stabilizer'' group is not commuting, and thus has no code space. In other words, one can not construct a quantum code by a classical linear code that is not dual-containing. In order to avoid this problem, Hsieh et al. \cite{C1} introduced a simple and fundamental class of quantum codes called entanglement-assisted quantum error correcting codes (EAQECCs). By relaxing the duality condition and using pre-shared entanglement between the sender and receiver, one can construct quantum codes from any classical linear codes. However, in general, the determination of the number of shared pairs that required to construct an EAQECC is not a easy thing. Guenda et al. \cite{EAQ} proved that this number can be evaluated by the dimension of the hull of classical linear codes. More precisely, given a classical linear code with the determined dimensional hull, one can obtain an EAQECC. For more details on EAQECCs, we refer the reader to \cite{QIP}. Therefore, the study of hull of linear codes is significant.

The purpose of this paper is to construct linear codes from generalized Reed-Solomon (GRS) codes or extended generalized Reed-Solomon codes and determine their hull. Inspired by the idea of \cite{Chenbc}, we propose several constructions of MDS codes with arbitrary dimensional hull. Furthermore, by using these MDS codes with arbitrary dimensional hull, we obtain several new infinite families of MDS EAQECCs.

This paper is organized as follows. In Section 2, we briefly recall some definitions and results about GRS codes and extended GRS codes. In Section 3, we present our constructions of MDS codes with arbitrary dimensional hull. In Section 4, we propose several families of MDS EAQECCs. In Section 5, we make a conclusion.

\section{Preliminaries}
In this section, we briefly recall some definitions and results about generalized Reed-Solomon codes, which will be employed in our discussion.

Let $q$ be a prime power and $\Ff_q$ denote the finite field with $q$ elements. We write $\Ff_q^*=\Ff_q\setminus\{0\}$. Assume that $\alpha_1,\alpha_2,\cdots,\alpha_n$ are $n$ distinct elements of $\Ff_q$, where $1<n\leq q$. For $n$ nonzero fixed elements $v_1,v_2,\cdots,v_n$ of $\Ff_q$ ($v_i$ may not be distinct), the GRS code associated with $\mathbf{a}=(\alpha_1,\alpha_2,\cdots,\alpha_n)$ and $\mathbf{v}=(v_1,v_2,\cdots,v_n)$ is defined as follows:
\begin{equation}\label{eq1}
GRS_k(\mathbf{a},\mathbf{v})=\{(v_1f(\alpha_1),v_2f(\alpha_2),\cdots,v_nf(\alpha_n)):f(x)\in\Ff_q[x],{\rm deg}(f(x))\leq k-1\}.
\end{equation}
A generator matrix of $GRS_k(\mathbf{a},\mathbf{v})$ is given by
$$
G=\left(\begin{matrix}
v_1 & v_2 &\cdots & v_n&\\
v_1\alpha_1 & v_2\alpha_2 &\cdots & v_n\alpha_n&\\
v_1\alpha_1^2 & v_2\alpha_2^2 &\cdots & v_n\alpha_n^2&\\
\vdots& \vdots &  \ddots & \vdots&\\
v_1\alpha_1^{k-1} & v_2\alpha_2^{k-1} &\cdots & v_n\alpha_n^{k-1}&\\
\end{matrix}\right).
$$
It is well known that the code $GRS_k(\mathbf{a},\mathbf{v})$ is a $q$-ary $[n,k,n-k+1]$-MDS code \cite[Th. 9.1.4]{Mac} and the dual of a GRS code is again a GRS code. More specifically,
$$
GRS_k(\mathbf{a},\mathbf{v})^{\bot}=GRS_{n-k}(\mathbf{a},\mathbf{v}^{\prime})
$$
for some $\mathbf{v}^{\prime}=(v_1^{\prime},v_2^{\prime},\cdots,v_n^{\prime})$ such that $v_i^{\prime}\neq 0$ for any $1\leq i\leq n$.

It is obvious that GRS codes exist for any length $n\leq q$ and any dimension $k\leq n$. GRS codes of length $q$ can be extended to the extended generalized Reed-Solomon codes with length $q+1$. Precisely speaking, the extended GRS code of length $q+1$ associated with $\mathbf{a}=(\alpha_1,\alpha_2,\cdots,\alpha_q)$ and $\mathbf{v}=(v_1,v_2,\cdots,v_q)$ is defined by
\begin{equation}\label{eq2}
GRS_k(\mathbf{a},\mathbf{v},\infty)=\{(v_1f(\alpha_1),v_2f(\alpha_2),\cdots,v_qf(\alpha_q),f_{k-1}):f(x)\in\Ff_q[x],{\rm deg}(f(x))\leq k-1\},
\end{equation}
where $\Ff_q=\{\alpha_1,\alpha_2,\cdots,\alpha_q\}$, $v_i\neq 0$ for all $1\leq i\leq q$ and $f_{k-1}$ is the coefficient of $x^{k-1}$ in $f(x)$. The extended GRS codes preserve the MDS property and $GRS_k(\mathbf{a},\mathbf{v},\infty)$ is a $q$-ary $[q+1,k,q-k+2]$ MDS code \cite{Mac}. A generator matrix of $GRS_k(\mathbf{a},\mathbf{v},\infty)$ is
$$
G=\left(\begin{matrix}
v_1 & v_2 &\cdots & v_n&0&\\
v_1\alpha_1 & v_2\alpha_2 &\cdots & v_n\alpha_n&0&\\
v_1\alpha_1^2 & v_2\alpha_2^2 &\cdots & v_n\alpha_n^2&0&\\
\vdots& \vdots &  \ddots & \vdots&\vdots&\\
v_1\alpha_1^{k-1} & v_2\alpha_2^{k-1} &\cdots & v_n\alpha_n^{k-1}&1&\\
\end{matrix}\right).
$$

In order to determine the hull of a GRS code, we need the following two lemmas.
\begin{lem}\label{lem1}{\rm \cite{Chenbc}}
Assume that $GRS_k(\mathbf{a},\mathbf{v})$ defined by $(\ref{eq1})$ is the GRS code associated with $\mathbf{a}$ and $\mathbf{v}$. For a codeword $\mathbf{c}=(v_1f(\alpha_1),v_2f(\alpha_2),\cdots,v_nf(\alpha_n))$ of $GRS_k(\mathbf{a},\mathbf{v})$, $\mathbf{c}$ is contained in $GRS_k(\mathbf{a},\mathbf{v})^{\bot}$ if and only if there is a polynomial $g(x)\in\Ff_q[x]$ with deg$(g(x))\leq n-k-1$ such that
$$
(v_1^2f(\alpha_1),v_2^2f(\alpha_2),\cdots,v_n^2f(\alpha_n))=(u_1g(\alpha_1),u_2g(\alpha_2),\cdots,u_ng(\alpha_n)),
$$
where $u_i=\prod_{1\leq j\leq n,j\neq i}(\alpha_i-\alpha_j)^{-1}$ for $1\leq i\leq n$.
\end{lem}

\begin{lem}\label{lem2}{\rm \cite{Chenbc}}
Assume that $GRS_k(\mathbf{a},\mathbf{v},\infty)$ defined by $(\ref{eq2})$ is the extended GRS code relative to $\mathbf{a}$ and $\mathbf{v}$. For a codeword $\mathbf{c}=(v_1f(\alpha_1),v_2f(\alpha_2),\cdots,v_qf(\alpha_q),f_{k-1})$ of $GRS_k(\mathbf{a},\mathbf{v},\infty)$, $\mathbf{c}$ is contained in $GRS_k(\mathbf{a},\mathbf{v},,\infty)^{\bot}$ if and only if there is a polynomial $g(x)\in\Ff_q[x]$ with deg$(g(x))\leq q-k$ such that
$$
(v_1^2f(\alpha_1),v_2^2f(\alpha_2),\cdots,v_q^2f(\alpha_q),f_{k-1})=(g(\alpha_1),g(\alpha_2),\cdots,g(\alpha_n),g_{q-k}),
$$
where $g_{q-k}$ stands for the coefficient of $x^{q-k}$ of $g(x)$.
\end{lem}

\section{Constructions of MDS codes with arbitrary dimensional hull}
In this section, utilizing generalized Reed-Solomon codes, we provide several constructions of MDS codes with arbitrary dimensional hull. As pointed out in the first section, the existence question about MDS codes with complementary duals over $\Ff_q$ has been completely addressed in \cite{Carlet} and \cite{Jin}. Here we only consider the case that the dimension of hull is greater than 0. By the definition of GRS codes, without loss of generality, we always restrict ourself to $k$-dimensional codes of length $n$ with $1<k\leq \lfloor n/2\rfloor$ in the sequel.

We first construct MDS codes with arbitrary dimensional hull over a finite field of even characteristic.

\begin{thm}\label{thm1}
Let $m>1$ be an integer and $q=2^m$. If $1<n\leq q$, then there exists a binary $[n,k]$ MDS code with $l$-dimensional hull for any $1\leq l \leq k$.
\end{thm}
\begin{proof}
We first show the existence of MDS codes with $l$-dimensional hull for any $1\leq l \leq k-1$.

Suppose that $\alpha_1,\alpha_2,\cdots,\alpha_n$ are $n$ distinct elements of $\Ff_q$, where $1<n\leq q$. Put $u_i=\prod_{1\leq j\leq n,j\neq i}(\alpha_i-\alpha_j)^{-1}$ for $1\leq i\leq n$. By the Frobenius transform of $\Ff_q$, for each $u_i$, there exists a unique element $v_i$ of $\Ff_q$ such that $v_i^2=u_i$, where $1\leq i\leq n$. Let $a\neq 1\in\Ff_q^*$ and $s$ be an integer with $1\leq s\leq k-1$. Take $\mathbf{a}=(\alpha_1,\alpha_2,\cdots,\alpha_n)$ and $\mathbf{v}=(av_1,av_2,\cdots,av_s,v_{s+1},\cdots,v_n)$. Then we obtain a $q$-ary GRS code $GRS_k(\mathbf{a},\mathbf{v})$ of length $n$ associated with $\mathbf{a}$ and $\mathbf{v}$ as follows
$$
GRS_k(\mathbf{a},\mathbf{v})=\{(av_1f(\alpha_1),\cdots,av_sf(\alpha_s),v_{s+1}f(\alpha_{s+1}),\cdots,v_nf(\alpha_n)):f(x)\in\Ff_q[x],{\rm deg}(f(x))\leq k-1\}.
$$
Assume that
$$
(av_1f(\alpha_1),\cdots,av_sf(\alpha_s),v_{s+1}f(\alpha_{s+1}),\cdots,v_nf(\alpha_n))
$$
is an arbitrary element of $GRS_k(\mathbf{a},\mathbf{v})\cap GRS_k(\mathbf{a},\mathbf{v})^{\bot}$. It follows from Lemma \ref{lem1} that there is a polynomial $g(x)\in\Ff_q[x]$ with deg$(g(x))\leq n-k-1$ such that
$$
(a^2v_1^2f(\alpha_1),\cdots,a^2v_s^2f(\alpha_2),v_{s+1}^2f(\alpha_{s+1}),\cdots,v_n^2f(\alpha_n))=(u_1g(\alpha_1),u_2g(\alpha_2),\cdots,u_ng(\alpha_n)).
$$
Note that $v_i^2=u_i$ for any $1\leq i\leq n$. Hence,
\begin{equation}\label{eq11}
(a^2u_1f(\alpha_1),\cdots,a^2u_sf(\alpha_2),u_{s+1}f(\alpha_{s+1}),\cdots,u_nf(\alpha_n))=(u_1g(\alpha_1),u_2g(\alpha_2),\cdots,u_ng(\alpha_n)).
\end{equation}
The last $n-s$ coordinates of $(\ref{eq11})$ imply that $f(\alpha_i)=g(\alpha_i)$ for $s<i\leq n$. Due to $k\leq \lfloor n/2\rfloor$, we get that deg$(f(x))\leq k-1\leq n-k-1$ and deg$(g(x))\leq n-k-1$. Since $1\leq s\leq k-1$, we have $f(x)=g(x)$ for any $x\in\Ff_q$. Considering the first $s$ coordinates of $(\ref{eq11})$, we deduce that
$$
a^2u_if(\alpha_i)=a^2u_ig(\alpha_i)=u_ig(\alpha_i),
$$
for any $1\leq i\leq s$. It follows from $a\neq 1$ that $g(\alpha_i)=0$. Precisely speaking, $g(x)$ has at least $s$ distinct roots. Observe that deg$(g(x))\leq k-1$ which implies that
$$
g(x)=h(x)\prod_{i=1}^s(x-\alpha_i),\ \ h(x)\in\Ff_q[x],\ \  {\rm deg}(h(x))\leq k-1-s.
$$
For any $g(x)\in\Ff_q[x]$ of the form $g(x)=h(x)\prod_{i=1}^s(x-\alpha_i)$, where ${\rm deg}(h(x))\leq k-1-s$, there exists a $f(x)=g(x)=h(x)\prod_{i=1}^s(x-\alpha_i)$ such that
$$
(a^2v_1^2f(\alpha_1),\cdots,a^2v_s^2f(\alpha_2),v_{s+1}^2f(\alpha_{s+1}),\cdots,v_n^2f(\alpha_n))=(u_1g(\alpha_1),u_2g(\alpha_2),\cdots,u_ng(\alpha_n)),
$$
which implies that
$$
(av_1f(\alpha_1),\cdots,av_sf(\alpha_s),v_{s+1}f(\alpha_{s+1}),\cdots,v_nf(\alpha_n))\in GRS_k(\mathbf{a},\mathbf{v})\cap GRS_k(\mathbf{a},\mathbf{v})^{\bot}.
$$
Therefore, the dimension of Hull$(GRS_k(\mathbf{a},\mathbf{v}))$ is $k-s$.

Next, we prove that there is a binary $[n,k]$ MDS code with $k$-dimensional hull. Let the symbols be the same as above. Take $\mathbf{a}=(\alpha_1,\alpha_2,\cdots,\alpha_n)$ and $\mathbf{v}=(v_1,v_2,\cdots,v_n)$. Consider the GRS code $GRS_k(\mathbf{a},\mathbf{v})$ of length $n$ as follows
$$
GRS_k(\mathbf{a},\mathbf{v})=\{(v_1f(\alpha_1),\cdots,v_nf(\alpha_n)):f(x)\in\Ff_q[x],{\rm deg}(f(x))\leq k-1\}.
$$
By the method analogous to that used above, we can verify that dim$({\rm Hull}(GRS_k(\mathbf{a},\mathbf{v})))=k$.
\end{proof}

Using the extended GRS codes, we construct MDS codes of length $q+1$ with variable dimensional hull.
\begin{thm}\label{thm7}
Let $q>3$ be an odd prime power. Then there exists a $q$-ary $[q+1,k]$ MDS code with $l$-dimensional hull for any $1\leq l \leq k-1$ and there exists a $q$-ary $[q+1,(q+1)/2]$ MDS code with $l$-dimensional hull for any $1\leq l \leq (q+1)/2$.
\end{thm}
\begin{proof}
Let $\Ff_q=\{\alpha_1,\alpha_2,\cdots,\alpha_q\}$ and $s$ an integer with $s\geq 1$. Assume that $\mathbf{a}=(\alpha_1,\alpha_2,\cdots,\alpha_{q})$ and $\mathbf{v}=(v_1,v_2,\cdots,v_s,1,\cdots,1)$, where $v_i\in\Ff_q^*$ and $v_i^2\neq 1$ for any $1\leq i\leq s$. Then we can obtain the $q$-ary extended GRS code $GRS_k(\mathbf{a},\mathbf{v},\infty)$ of length $q+1$ associated with $\mathbf{a}$ and $\mathbf{v}$ as follows
$$
GRS_k(\mathbf{a},\mathbf{v},\infty)=\{(v_1f(\alpha_1),\cdots,v_sf(\alpha_s),f(\alpha_{s+1}),\cdots,f(\alpha_q),f_{k-1}):f(x)\in\Ff_q[x],{\rm deg}(f(x))\leq k-1\},
$$
where $f_{k-1}$ is the coefficient of $x^{k-1}$ in $f(x)$. Suppose that
$$
(v_1f(\alpha_1),\cdots,v_sf(\alpha_s),f(\alpha_{s+1}),\cdots,f(\alpha_q),f_{k-1})
$$
is an arbitrary element of $GRS_k(\mathbf{a},\mathbf{v},\infty)\cap GRS_k(\mathbf{a},\mathbf{v},\infty)^{\bot}$. It follows from Lemma \ref{lem2} that there exists a polynomial $g(x)\in\Ff_q[x]$ with deg$(g(x))\leq q-k$ such that
\begin{equation}\label{eq71}
(v_1^2f(\alpha_1),\cdots, v_s^2f(\alpha_{s}),f(\alpha_{s+1}),\cdots,f(\alpha_{q}),f_{k-1})=(g(\alpha_1),\cdots,g(\alpha_q),g_{q-k}),
\end{equation}
where $g_{q-k}$ stands for the coefficient of $x^{q-k}$ of $g(x)$. We divide the rest of the proof into three cases.

Case 1): $1<k<\frac{q+1}{2}$ and $s\leq k-2$. From the $q+1-s$ last coordinates of $(\ref{eq71})$, we have that $f(\alpha_{i})=g(\alpha_i)$ for all $s<i\leq q$ and $f_{k-1}=g_{q-k}$. Note that deg$(f(x))\leq k-1< q-k$, deg$(g(x))\leq q-k$ and $q-s\geq q-k+2$. Since $f(x)=g(x)$ has at least $q-s$ distinct roots, $f(x)=g(x)$ for any $x\in\Ff_q$. It follows from $f_{k-1}=g_{q-k}$ that $f_{k-1}=0$; otherwise we would have $k=\frac{q+1}{2}$, which leads a contradiction. Hence, deg$(g(x))\leq k-2$. According to the first $s$ coordinates of $(\ref{eq71})$, we deduce that
$$
v_i^2f(\alpha_i)=v_i^2g(\alpha_i)=g(\alpha_i)
$$
for $1\leq i\leq s$. Due to $v_i^2\neq1$, we get that $g(\alpha^i)=0$ for any $1\leq i\leq s$. In other words, $g(x)$ has at least $s$ distinct roots. Since $f(x)=g(x)$ and deg$(g(x))\leq k-2$, we obtain
$$
g(x)=h(x)\prod_{i=1}^s(x-\alpha^i),\ \ h(x)\in\Ff_q[x],\ \  {\rm deg}(h(x))\leq k-2-s.
$$
For any $g(x)\in\Ff_q[x]$ of the form $g(x)=h(x)\prod_{i=1}^s(x-\alpha_i)$, where ${\rm deg}(h(x))\leq k-2-s$, there exists a $f(x)=g(x)=h(x)\prod_{i=1}^s(x-\alpha_i)$ such that
$$
(v_1^2f(\alpha_1),\cdots, v_s^2f(\alpha_{s}),f(\alpha_{s+1}),\cdots,f(\alpha_{q}),f_{k-1})=(g(\alpha_1),\cdots,g(\alpha_q),g_{q-k})
$$
which implies that
$$
(v_1f(\alpha_1),\cdots,v_sf(\alpha_s),f(\alpha_{s+1}),\cdots,f(\alpha_q),f_{k-1})\in GRS_k(\mathbf{a},\mathbf{v},\infty)\cap GRS_k(\mathbf{a},\mathbf{v},\infty)^{\bot}.
$$
Therefore, dim$({\rm Hull}(GRS_k(\mathbf{a},\mathbf{v},\infty)))=k-1-s$.

Case 2): $k=\frac{q+1}{2}$ and $s\leq k-1$. Equation $(\ref{eq71})$ yields $f_{k-1}=g_{q-k}$. Hence, one has deg$(f(x)-g(x))\leq k-2$. Using the same argument as in the proof of Case 1, we can easily carry out dim$({\rm Hull}(GRS_k(\mathbf{a},\mathbf{v},\infty)))=k-s$.

Case 3): Let the symbols be the same as above. Take $\mathbf{a}=(\alpha_1,\alpha_2,\cdots,\alpha_{q})$ and $\mathbf{v}=(1,1,\cdots,1)$. Consider the extended GRS code $GRS_k(\mathbf{a},\mathbf{v},\infty)$ of length $n$ as follows
$$
GRS_k(\mathbf{a},\mathbf{v},\infty)=\{(f(\alpha_1),\cdots,f(\alpha_q),f_{k-1}):f(x)\in\Ff_q[x],{\rm deg}(f(x))\leq k-1\},
$$
By the method analogous to that used in Case 1 and Case 2, we can verify that
$$
{\rm dim}({\rm Hull}(GRS_k(\mathbf{a},\mathbf{v},\infty)))=k-1,\ \ \ {\rm dim}({\rm Hull}(GRS_k(\mathbf{a},\mathbf{v},\infty)))=k,
$$
respectively.
\end{proof}

Proceeding as in the proof of Case 1 of Theorem \ref{thm7}, we have the following result over a finite field of even characteristic. Here, we omit the proof of Theorem \ref{thm8}.

\begin{thm}\label{thm8}
Let $q=2^m$, where $m>1$ is an integer. Then there exists a $q$-ary $[q+1,k]$ MDS code with $l$-dimensional hull for any $1\leq l \leq k-1$.
\end{thm}

In the following, we consider GRS codes over a finite field with odd characteristic. By taking the set $\{\alpha_1,\alpha_2,\cdots,\alpha_n\}$ defined by (\ref{eq1}) as an multiplicative subgroup of $\Ff_q^*$, we obtain the following GRS codes with variable dimensional hull.

\begin{thm}\label{thm2}
Let $q>3$ be an odd prime power. If $n>1$ with $n|(q-1)$, then there exists a $q$-ary $[n,k]$ MDS code with $l$-dimensional hull for any $1\leq l \leq k-1$.
\end{thm}
\begin{proof}
Since $n>1$ and $n|(q-1)$, there exists an multiplicative subgroup $G$ of $\Ff_q^*$ of order $n$. Let $\alpha$ be a generator of $G$ and $s$ an integer with $1\leq s\leq k-2$. Put $\mathbf{a}=(\alpha,\alpha^2,\cdots,\alpha^{n})$ and $\mathbf{v}=(v_1,v_2,\cdots,v_s,1,\cdots,1)$, where $v_i\in\Ff_q^*$ and $v_i^2\neq 1$ for any $1\leq i\leq s$. Consider the $q$-ary GRS code $\cc$ of length $n$ associated with $\mathbf{a}$ and $\mathbf{v}$ as follows
$$
\cc=\{(v_1f(\alpha),\cdots, v_sf(\alpha^{s}),f(\alpha^{s+1}),\cdots,f(\alpha^{n})):f(x)\in\Ff_q[x],{\rm deg}(f(x))\leq k-1\}.
$$
Let $(v_1f(\alpha),\cdots, v_sf(\alpha^{s}),f(\alpha^{s+1}),\cdots,f(\alpha^{n}))\in\cc\cap \cc^{\bot}$. By Lemma \ref{lem1}, there is a polynomial $g(x)\in\Ff_q[x]$ with deg$(g(x))\leq n-k-1$ such that
\begin{equation}\label{eq21}
(v_1^2f(\alpha),\cdots, v_s^2f(\alpha^{s}),f(\alpha^{s+1}),\cdots,f(\alpha^{n}))=(u_1g(\alpha ),u_2g(\alpha^2),\cdots,u_ng(\alpha^n)),
\end{equation}
where $u_i=\prod_{1\leq j\leq n,j\neq i}(\alpha^i-\alpha^j)^{-1}$ for $1\leq i\leq n$. It is easy to see that
$$
u_i=\prod_{1\leq j\leq n,j\neq i}(\alpha^i-\alpha^j)^{-1}=(\alpha^{i(n-1)})^{-1}\prod_{1\leq j\leq n,j\neq i}(1-\alpha^{j-i})^{-1}=\alpha^{i}\prod_{1\leq j\leq n-1}(1-\alpha^{j})^{-1}.
$$
Note that $\prod_{1\leq j\leq n-1}(x-\alpha^{j})=\frac{x^n-1}{x-1}=1+x+\cdots+x^{n-1}$. Then we obtain $u_i=\alpha^{i}n^{-1}$ due to ${\rm gcd}(n,q)=1$.

The last $n-s$ coordinates of $(\ref{eq21})$ give $f(\alpha^i)=\alpha^{i}n^{-1}g(\alpha^i)$ for $s<i\leq n$. Thanks to $k\leq \lfloor n/2\rfloor$, we have that deg$(f(x))\leq k-1\leq n-k-1$ and deg$(g(x))\leq n-k-1$. Since $1\leq s\leq k-2$, we obtain $f(x)=n^{-1}xg(x)$ for any $x\in\Ff_q$. From the first $s$ coordinates of $(\ref{eq21})$, we derive that
$$
v_i^2f(\alpha^i)=v_i^2n^{-1}\alpha^ig(\alpha^i)=n^{-1}\alpha^ig(\alpha^i),
$$
for any $1\leq i\leq s$. It follows from $v_i^2\neq1$ that $g(\alpha^i)=0$ for any $1\leq i\leq s$. In other words, $g(x)$ has at least $s$ distinct roots. Since $f(x)=n^{-1}xg(x)$ and deg$(f(x))\leq k-1$, we get deg$(g(x))\leq k-2$ and
$$
g(x)=h(x)\prod_{i=1}^s(x-\alpha^i),\ \ h(x)\in\Ff_q[x],\ \  {\rm deg}(h(x))\leq k-2-s.
$$
Therefore, dim$({\rm Hull}(\cc))=k-1-s$.

Below, we show that there exists the MDS code $\cc$ with dim$({\rm Hull}(\cc))=k-1$. Let the symbols be the same as above. Take $\mathbf{a}=(\alpha ,\alpha^2,\cdots,\alpha^n)$ and $\mathbf{v}=(1,1,\cdots,1)$. Consider the GRS code $GRS_k(\mathbf{a},\mathbf{v})$ of length $n$ as follows
$$
GRS_k(\mathbf{a},\mathbf{v})=\{(f(\alpha ),\cdots,f(\alpha^n)):f(x)\in\Ff_q[x],{\rm deg}(f(x))\leq k-1\}.
$$
Similarly, we can verify that dim$({\rm Hull}(\cc))=k-1$.
\end{proof}

Adding the zero element into the multiplicative subgroup $G$ defined in Theorem \ref{thm2}, a $q$-ary $[n+1,k]$ GRS code is given as follows.

\begin{thm}\label{thm3}
Let $q>3$ be an odd prime power. Assume that $n>1$ with $n|(q-1)$. If $-n$ is a square of $\Ff_q$, then there exists a $q$-ary $[n+1,k]$ MDS code with $l$-dimensional hull for any $1\leq l \leq k$.
\end{thm}
\begin{proof}
Note that $n>1$ and $n|(q-1)$. Clearly, there exists an multiplicative subgroup $G$ of $\Ff_q^*$ of order $n$. Let $\alpha$ be a generator of $G$ and $s$ an integer with $1\leq s\leq k-1$. Since ${\rm gcd}(n,q)=1$ and $-n$ is a square of $\Ff_q$, there exists an nonzero element $a$ such that $a^2=-n^{-1}$. Let $\mathbf{a}=(\alpha,\alpha^2,\cdots,\alpha^{n},0)$ and $\mathbf{v}=(av_1,av_2,\cdots,av_s,a,\cdots,a,1)$, where $v_i\in\Ff_q^*$ and $v_i^2\neq 1$ for any $1\leq i\leq s$. Consider the $q$-ary GRS code $\cc$ of length $n$ associated with $\mathbf{a}$ and $\mathbf{v}$ as follows
$$
\cc=\{(av_1f(\alpha),\cdots, av_sf(\alpha^{s}),af(\alpha^{s+1}),\cdots,af(\alpha^{n}),f(0)):f(x)\in\Ff_q[x],{\rm deg}(f(x))\leq k-1\}.
$$
Let $(av_1f(\alpha),\cdots, av_sf(\alpha^{s}),af(\alpha^{s+1}),\cdots,af(\alpha^{n}),f(0))$ be an arbitrary element of $\cc\cap \cc^{\bot}$. Using Lemma \ref{lem1}, there is a polynomial $g(x)\in\Ff_q[x]$ with deg$(g(x))\leq n-k$ such that
\begin{eqnarray}\label{eq31}
&&(-n^{-1}v_1^2f(\alpha),\cdots, -n^{-1}v_s^2f(\alpha^{s}),-n^{-1}f(\alpha^{s+1}),\cdots,-n^{-1}f(\alpha^{n}),f(0))\nonumber \\
&=&(u_1g(\alpha),u_2g(\alpha^2),\cdots,u_ng(\alpha^n),u_{n+1}g(0)),
\end{eqnarray}
where $u_i=\alpha^{-i}\prod_{1\leq j\leq n,j\neq i}(\alpha^i-\alpha^j)^{-1}$ for $1\leq i\leq n$ and $u_{n+1}=\prod_{j=1}^n(0-\alpha^j)^{-1}$. From the proof of Theorem \ref{thm2}, we have $u_i=n^{-1}$ for $1\leq i\leq n$. It is easy to check that $u_{n+1}=-1$.

It follows from the last $n+1-s$ coordinates of $(\ref{eq31})$ that $f(0)=-g(0)$ and $f(\alpha^i)=-g(\alpha^i)$ for $s<i\leq n$, i.e., $f(x)=-g(x)$ has at least $n-s+1$ distinct roots. Now that $1\leq s\leq k-1$, deg$(f(x))\leq k-1\leq n-k-1$ and deg$(g(x))\leq n-k-1$. Hence, $f(x)=-g(x)$ for any $x\in\Ff_q$. By the first $s$ coordinates of $(\ref{eq31})$, we obtain
$$
-n^{-1}v_1^2f(\alpha^i)=n^{-1}v_1^2g(\alpha^i)=n^{-1}g(\alpha^i),
$$
for any $1\leq i\leq s$. Due to $v_i^2\neq 1$, we have $g(\alpha^i)=0$ for $1\leq i\leq s$, i.e., $g(x)$ has at least $s$ distinct roots. Note that deg$(g(x))\leq k-1$ which implies that
$$
g(x)=h(x)\prod_{i=1}^s(x-\alpha^i),\ \ h(x)\in\Ff_q[x],\ \  {\rm deg}(h(x))\leq k-1-s.
$$
Therefore, the dimension of Hull$(GRS_k(\mathbf{a},\mathbf{v}))$ is $k-s$.

Below, we verify that $\cc\cap \cc^{\bot}=\cc$. Let the symbols be the same as above. Take $\mathbf{a}=(\alpha,\alpha^2,\cdots,\alpha^{n},0)$ and $\mathbf{v}=(a,\cdots,a,1)$. Consider the GRS code $\cc$ of length $n$ as follows
$$
\cc=\{(af(\alpha^1),\cdots,af(\alpha^n),f(0)):f(x)\in\Ff_q[x],{\rm deg}(f(x))\leq k-1\}.
$$
An argument similar to the one used above shows that dim$({\rm Hull}(\cc))=k$.
\end{proof}

Next, we provide a construction of MDS codes with variable dimensional hull from GRS codes of even length.

\begin{thm}\label{thm4}
Let $q\equiv1\pmod {4}$ be an odd prime power. Assume that $m>1$ is an integer such that $m|(q-1)$. If $n=2m<q-1$, then there exists a $q$-ary $[n,k]$ MDS code with $l$-dimensional hull for any $1\leq l \leq k-1$.
\end{thm}
\begin{proof}
It follows from $m|(q-1)$ that there exists an multiplicative subgroup $G$ of $\Ff_q^*$ of order $m$. Let $\alpha$ be a generator of $G$ and $s$ an integer with $1\leq s\leq k-2$. Since $2m<q-1$, we can take a square element $\omega\in\Ff_q^*\setminus G$. Let $\omega=a^2$ and $\gamma$ be a primitive element of $\Ff_q^*$. Set $\mathbf{a}=(\alpha,\cdots,\alpha^{m},\omega\alpha,\cdots,\omega\alpha^{m})$ and $\mathbf{v}=(v_1,v_2,\cdots,v_s,1,\cdots,1,\underbrace{\gamma^{\frac{q-1}{4}}a^{1-m},\cdots,\gamma^{\frac{q-1}{4}}a^{1-m}}_{m})$, where $v_i\in\Ff_q^*$ and $v_i^2\neq 1$ for any $1\leq i\leq s$. Consider the $q$-ary GRS code $\cc$ of length $n$ relative to $\mathbf{a}$ and $\mathbf{v}$ as follows
\begin{eqnarray*}
\cc&=&\{(v_1f(\alpha),\cdots, v_sf(\alpha^{s}),f(\alpha^{s+1}),\cdots,f(\alpha^{m}),\gamma^{\frac{q-1}{4}}a^{1-m}f(\omega\alpha),\cdots,\gamma^{\frac{q-1}{4}}a^{1-m}f(\omega\alpha^m)):\\
&&f(x)\in\Ff_q[x],{\rm deg}(f(x))\leq k-1\}.
\end{eqnarray*}
Suppose that
$$
(v_1f(\alpha),\cdots, v_sf(\alpha^{s}),f(\alpha^{s+1}),\cdots,f(\alpha^{m}),\gamma^{\frac{q-1}{4}}a^{1-m}f(\omega\alpha),\cdots,\gamma^{\frac{q-1}{4}}a^{1-m}f(\omega\alpha^m))
$$
is an arbitrary element of $\cc\cap \cc^{\bot}$. By Lemma \ref{lem1}, there is a polynomial $g(x)\in\Ff_q[x]$ with deg$(g(x))\leq n-k-1$ such that
\begin{eqnarray*}\label{eq41}
&&(v_1^2f(\alpha),\cdots, v_s^2f(\alpha^{s}),f(\alpha^{s+1}),\cdots,f(\alpha^{m}),-\omega^{1-m}f(\omega\alpha),\cdots,-\omega^{1-m}f(\omega\alpha^m))\nonumber \\
&=&(u_1g(\alpha),u_2g(\alpha^2),\cdots,u_mg(\alpha^m),u_{m+1}g(\omega\alpha),\cdots,u_{2m}g(\omega\alpha^m)),
\end{eqnarray*}
where
$$
u_i=\prod_{1\leq j\leq m,j\neq i}(\alpha^i-\alpha^j)^{-1}\prod_{h=1}^m(\alpha^i-\omega\alpha^h)^{-1}
$$
and
$$
u_{m+i}=\prod_{h=1}^m(\omega\alpha^i-\alpha^h)^{-1}\prod_{1\leq j\leq m,j\neq i}(\omega\alpha^i-\omega\alpha^j)^{-1},
$$
for $1\leq i\leq m$. Since $x^m-b^m=\prod_{h=1}^m(x-b\alpha^h)$ for any $b\in\Ff_q$, from the proof of Theorem \ref{thm2}, we get that
$$
u_i=m^{-1}\alpha^i(1-\omega^m)^{-1}\ \ \ {\rm and}\ \ u_{m+i}=-\omega^{1-m}m^{-1}\alpha^i(1-\omega^m)^{-1},
$$
where $1\leq i\leq m$. Proceeding as in the proof of Theorem \ref{thm1}, we complete the proof of this theorem.
\end{proof}

Exchanging the multiplicative subgroup of $\Ff_q^*$ in Theorem \ref{thm2} by the additive subgroup of $\Ff_q$, we have the following result.
\begin{thm}\label{thm5}
Let $q>3$ be an odd prime power. If $n>1$ with $n|q$, then there exists a $q$-ary $[n,k]$ MDS code with $l$-dimensional hull for any $1\leq l \leq k$.
\end{thm}
\begin{proof}
Let $G$ be an additive subgroup of $\Ff_q$ of order $n$ and $s$ an integer with $1\leq s\leq k-1$. Label the elements of $G=\{\alpha_1,\cdots,\alpha_n\}$. Put $\mathbf{a}=(\alpha_1,\alpha_2,\cdots,\alpha_{n})$ and $\mathbf{v}=(v_1,v_2,\cdots,v_s,1,\cdots,1)$, where $v_i\in\Ff_q^*$ and $v_i^2\neq 1$ for any $1\leq i\leq s$. Consider the $q$-ary GRS code $\cc$ of length $n$ associated with $\mathbf{a}$ and $\mathbf{v}$ as follows
$$
\cc=\{(v_1f(\alpha_1),\cdots, v_sf(\alpha_{s}),f(\alpha_{s+1}),\cdots,f(\alpha_{n})):f(x)\in\Ff_q[x],{\rm deg}(f(x))\leq k-1\}.
$$
Let $(v_1f(\alpha_1),\cdots, v_sf(\alpha_{s}),f(\alpha_{s+1}),\cdots,f(\alpha_{n}))\in\cc\cap \cc^{\bot}$. It follows from Lemma \ref{lem1} that there is a polynomial $g(x)\in\Ff_q[x]$ with deg$(g(x))\leq n-k-1$ such that
\begin{equation}\label{eq51}
(v_1^2f(\alpha_1),\cdots, v_s^2f(\alpha_{s}),f(\alpha_{s+1}),\cdots,f(\alpha_{n}))=(u_1g(\alpha_1),u_2g(\alpha_2),\cdots,u_ng(\alpha_n)),
\end{equation}
where $u_i=\prod_{1\leq j\leq n,j\neq i}(\alpha_i-\alpha_j)^{-1}$ for $1\leq i\leq n$. It can be easily seen that $u_i=\prod_{\omega\in G\setminus\{0\}}\omega^{-1}$ for each $1\leq i\leq n$. The desired result can be obtained by using the same argument as in the proof of Theorem \ref{thm1}.
\end{proof}

\begin{thm}\label{thm6}
Let $q\equiv1\pmod {4}$ be an odd prime power. Suppose that $m>1$ is an integer such that $m|q$. If $n=2m<q$, then there exists a $q$-ary $[n,k]$ MDS code with $l$-dimensional hull for any $1\leq l \leq k$.
\end{thm}
\begin{proof}
Since $m>1$ and $m|q$, there exists an additive subgroup $G$ of $\Ff_q$ of order $m$. Let $a\in\Ff_q\setminus G$ and $a+G$ be the coset of $a$ relative to $G$. Label the elements of $G=\{\alpha_1,\cdots,\alpha_n\}$. Suppose that $\gamma$ is a primitive element of $\Ff_q^*$. Put $\mathbf{a}=(\alpha_1,\cdots,\alpha_{m},a+\alpha_1,\cdots,a+\alpha_{m})$ and $\mathbf{v}=(\gamma^{\frac{q-1}{4}}v_1,\gamma^{\frac{q-1}{4}}v_2,\cdots,\gamma^{\frac{q-1}{4}}v_s,\gamma^{\frac{q-1}{4}},\cdots,\gamma^{\frac{q-1}{4}},\underbrace{1,\cdots,1}_{m})$, where $v_i\in\Ff_q^*$ and $v_i^2\neq 1$ for any $1\leq i\leq s$. Consider the $q$-ary GRS code $\cc$ of length $n$ associated with $\mathbf{a}$ and $\mathbf{v}$ as follows
\begin{eqnarray*}
\cc&=&\{(\gamma^{\frac{q-1}{4}}v_1f(\alpha_1),\cdots, \gamma^{\frac{q-1}{4}}v_sf(\alpha_{s}),\gamma^{\frac{q-1}{4}}f(\alpha_{s+1}),\cdots,\gamma^{\frac{q-1}{4}}f(\alpha_{m}),f(a+\alpha_1),\cdots,f(a+\alpha_m)):\\
&&f(x)\in\Ff_q[x],{\rm deg}(f(x))\leq k-1\}.
\end{eqnarray*}
Suppose that
$$
(\gamma^{\frac{q-1}{4}}v_1f(\alpha_1),\cdots, \gamma^{\frac{q-1}{4}}v_sf(\alpha_{s}),\gamma^{\frac{q-1}{4}}f(\alpha_{s+1}),\cdots,\gamma^{\frac{q-1}{4}}f(\alpha_{m}),f(a+\alpha_1),\cdots,f(a+\alpha_m))
$$
is an arbitrary element of $\cc\cap \cc^{\bot}$. From Lemma \ref{lem1}, there is a polynomial $g(x)\in\Ff_q[x]$ with deg$(g(x))\leq n-k-1$ such that
\begin{eqnarray*}\label{eq61}
&&(-v_1^2f(\alpha_1),\cdots, -v_s^2f(\alpha_{s}),-f(\alpha_{s+1}),\cdots,-f(\alpha_{m}),f(a+\alpha_1),\cdots,f(a+\alpha_m))\nonumber \\
&=&(u_1g(\alpha_1),u_2g(\alpha_2),\cdots,u_mg(\alpha_m),u_{m+1}g(a+\alpha_1),\cdots,u_{2m}g(a+\alpha_m)),
\end{eqnarray*}
where
$$
u_i=-a^{-1}\prod_{1\leq j\leq n, j\neq i}(\alpha_i-\alpha_j)^{-1}(\alpha_i-a-\alpha_j)^{-1}
$$
and
$$
u_{m+i}=a^{-1}\prod_{1\leq j\leq n, j\neq i}(\alpha_i-\alpha_j)^{-1}(\alpha_i+a-\alpha_j)^{-1},
$$
for $1\leq i\leq m$. It can be easily checked that
\begin{eqnarray*}
u_i&=&-a^{-1}\prod_{1\leq j\leq n, j\neq i}(\alpha_i-\alpha_j)^{-1}(\alpha_i-a-\alpha_j)^{-1}\\
&=&-a^{-1}\prod_{\omega\in G\setminus\{0\}}\omega^{-1}(\omega-a)^{-1}\\
&=&-a^{-1}\prod_{\omega\in G\setminus\{0\}}(\omega^2-a\omega)^{-1}
\end{eqnarray*}
and
\begin{eqnarray*}
u_{m+i}&=&a^{-1}\prod_{1\leq j\leq n, j\neq i}(\alpha_i-\alpha_j)^{-1}(\alpha_i+a-\alpha_j)^{-1}\\
&=&a^{-1}\prod_{\omega\in G\setminus\{0\}}(\omega^2+a\omega)^{-1}\\
&=&a^{-1}\prod_{\omega\in G\setminus\{0\}}(\omega^2-a\omega)^{-1},
\end{eqnarray*}
where $1\leq i\leq m$. Proceeding as in the proof of Theorem \ref{thm1}, we complete the proof of this theorem.
\end{proof}

\section{Constructions of entanglement-assisted quantum error correcting codes}
In this section, we introduce some definitions and notations about entanglement-assisted quantum error-correcting codes (EAQECCs). We also present several infinite families of optimal EAQECCs. We start with some notations that will be used in the following.

 Assume that $\mathcal{H}^{\otimes n}$ is the tensor product Hilbert space corresponding to an $n$-qubit system. Let $\mathbf{B}=B_1\otimes B_2\otimes \cdots \otimes B_n$ be an $n$-qubit Pauli matrix, where $B_i\in \{I,X,Y,Z\}$ is an element of the set of Pauli matrices. Assume that $G^n$ is the group of all $4^n$ $n$-qubit Pauli matrices with all possible phases. Define an equivalent class by $[\mathbf{B}]=\{\alpha \mathbf{B}| \alpha\in\Cc, |\alpha|=1\}$. The set $[G^n]=\{[\mathbf{B}]: \mathbf{B}\in G^n\}$ is a commutative group under the multiplication
$$
[\mathbf{B}][\mathbf{C}]=[B_1C_1]\otimes\cdots \otimes [B_nC_n]=[\mathbf{B}\mathbf{C}].
$$

Let $(\Zz_2)^{2n}$ be the vector space of binary vectors with length $2n$. For any $\mathbf{v}\in(\Zz_2)^{2n}$, $\mathbf{v}$ is represented by $(\mathbf{x}|\mathbf{y})$, where $\mathbf{x},\mathbf{y}\in(\Zz_2)^n$. The map $N$ from $(\Zz_2)^{2n}$ to $G^n$ is defined by
$$
N_{\mathbf{v}}=N_{v_1}\otimes N_{v_2} \otimes \cdots \otimes N_{v_1}.
$$
Put
$$
X^{\mathbf{x}}=X^{x_1}\otimes \cdots \otimes X^{x_n},
$$
$$
Y^{\mathbf{y}}=Y^{y_1}\otimes \cdots \otimes Y^{y_n}.
$$
In the single qubit case, we have $[N_{(\mathbf{x}|\mathbf{y})}]=[X^{\mathbf{x}}Y^{\mathbf{y}}]$. In the qubit Hilbert space $\mathcal{H}$, let $\mathcal{L}$ stand for the space of linear operators. Consider the isometric operator $U:\mathcal{H}^{\otimes n_1}\rightarrow \mathcal{H}^{\otimes n_2}$ and its completely positive, trace preserving (CPTP) map $\widehat{U}$ from $\mathcal{L}^{\otimes n_1}$ to $\mathcal{L}^{\otimes n_2}$ is defined as $\widehat{U}(\rho)=U\rho U^{\dagger}$.

\begin{figure}
\centering
\includegraphics[width=5.77in,height=2.75in]{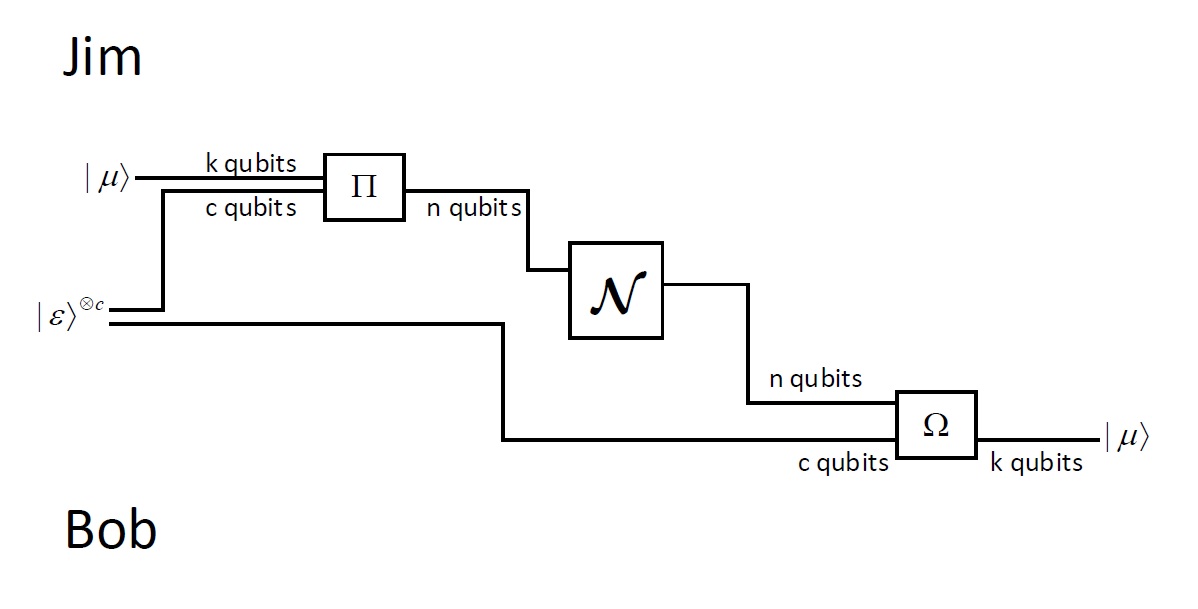}
\caption{A generic entanglement assisted quantum code}
\label{fig1}
\end{figure}

The following communication scenario showed in Fig. \ref{fig1}, contains two spatially separated parties, Jim and Bob. The schemes at their disposal are
\begin{enumerate}
\item A noisy channel, defined as a CPTP map $\mathcal{N}:\mathcal{L}^{\otimes n}\rightarrow\mathcal{L}^{\otimes n}$, takes density operators on Jim’s system to density operators on Bob’s system.
\item The $c$-ebit state $|\varepsilon\rangle^{\otimes c}$ shared between Bob and Jim.
\end{enumerate}
With the above schemes, Jim wants to send $k$-qubit quantum information to Bob perfectly. An $[[n, k, d; c]]_q$ EAQECC is composed of
\begin{enumerate}
\item An encoding isometry $\Pi:\mathcal{L}^{\otimes k}\otimes\mathcal{L}^{\otimes k}\rightarrow\mathcal{L}^{\otimes n}$.
\item A decoding CPTP map $\Omega:\mathcal{L}^{\otimes n}\otimes\mathcal{L}^{\otimes c}\rightarrow\mathcal{L}^{\otimes k}$.
\end{enumerate}
with $\Omega\circ\mathcal{N}\circ\Pi\circ U={\rm id}^{\otimes k}$, where $U$ is the isometry appending the state $|\varepsilon\rangle^{\otimes c}$, i.e.,
$
U|\mu\rangle=|\mu\rangle|\varepsilon\rangle^{\otimes c},
$
and ${\rm id}: \mathcal{L}\rightarrow\mathcal{L}$ is the identity map on a single qubit. The protocol spends $c$ ebits of entanglement and produces $k$ perfect qubit channels. The parameter $k-c$ is a good judgement of the net noiseless quantum resources gained. It is clear that the protocol is non-trivial if $k-c$ is negative.

The performance of EAQECCs is determined by its rate $\frac{k}{n}$ and net rate $\frac{k-c}{n}$. In general, the net rate can be positive, negative, or zero. If the net rate is negative, the corresponding EAQECC may have practical applications. EAQECCs with positive net rates can be employed in some other ways to increase the power and flexibility of quantum communications. Brun et al. \cite{Brun} indicated that it is possible to construct catalytic codes if the net rate of an EAQECC is positive.

In \cite{EAQ2}, Wilde and Brun provided a method for constructing EAQECCS by utilizing classical linear codes over finite fields as follows.
\begin{lem}{\rm \cite{EAQ2}}\label{lem41}
Assume that $H_1$ and $H_2$ are parity check matrices of two $q$-ary linear codes $[n, k_1, d_1]$ and $[n, k_2, d_2]$, respectively. Then there exists an $[[n,k_1+k_2-n+c,{\rm min}\{d_1,d_2\};c]]_q$ EAQECC, where $c=rank(H_1H_2^t)$ is the required number of maximally entangled states.
\end{lem}

Brun et al. \cite{C1} has given the Singleton bound for an EAQECC in the following lemma.
\begin{lem}{\rm \cite{C1}}\label{lem42}
For any $[[n,k,d;c]]_q$ EAQECC, it satisfies
$$
n+c-k\geq 2(d-1),
$$
where $0\leq c\leq n-1$.
\end{lem}

An EAQECC is called a MDS EAQECC if its parameters achieve the Singleton bound. In general, the parameter $c=rank(H_1H_2^t)$ is not easy to compute until Guenda et al. \cite{EAQ} provided a relation between the required number of maximally entangled states and the dimension of the hull of a classical code as follows.
\begin{lem}{\rm \cite{EAQ}}\label{lem43}
Let $\cc$ be a $q$-ary linear codes with $[n, k, d]$. Assume that $H$ is a parity check matrix and $G$ is a generator matrix of $\cc$. Then we have
$$
rank(HH^t)=n-k-{\rm dim}({\rm Hull}(\cc))=n-k-{\rm dim}({\rm Hull}(\cc^{\bot})),
$$
and
$$
rank(GG^t)=k-{\rm dim}({\rm Hull}(\cc))=k-{\rm dim}({\rm Hull}(\cc^{\bot})).
$$
\end{lem}

As a direct consequence of Lemma \ref{lem41}, \ref{lem42} and \ref{lem43}, one has the following lemma.

\begin{lem}{\rm \cite{EAQ}}\label{lem44}
Let $\cc$ be an $[n, k, d]$ a linear codes over $\Ff_q$ and $\cc^{\bot}$ its Euclidean dual with $[n, n-k, d^{\bot}]$. Then there exist $[[n,k-{\rm dim}({\rm Hull}(\cc)),d;n-k-{\rm dim}({\rm Hull}(\cc))]]_q$ and $[[n,n-k-{\rm dim}({\rm Hull}(\cc)),d^{\bot};k-{\rm dim}({\rm Hull}(\cc))]]_q$ EAQECCs. Moreover, if $\cc$ is MDS, then the two EAQECCs are also MDS.
\end{lem}

With the above lemma, the construction of MDS EAQECCs turns into that of MDS linear codes with the determined dimensional hull. In Section 3, using GRS codes and extended GRS codes, we presented several families of MDS codes and completely determined their hull. Let $1<k\leq \lfloor n/2\rfloor$ in the sequel. By Lemma \ref{lem44} and all the theorems of Section 3, we have the following results directly.

\begin{thm}\label{thm41}
Let $m>1$ be an integer and $q=2^m$. If $1<n\leq q$, then there exist $[[n,k-s,n-k+1;n-k-s]]_q$ and $[[n,n-k-s,k+1;k-s]]_q$ MDS EAQECCs, where $1\leq s \leq k$.
\end{thm}

\begin{thm}\label{thm44}
Let $q=2^m$, where $m>1$ is an integer. If $n=q+1$ and $1\leq s \leq k-1$, then there exist $[[n,k-s,n-k+1;n-k-s]]_q$ and $[[n,n-k-s,k+1;k-s]]_q$ MDS EAQECCs.
\end{thm}

\begin{thm}\label{thm42}
Let $q>3$ be an odd prime power. Then there exist $[[n,k-s,n-k+1;n-k-s]]_q$ and $[[n,n-k-s,k+1;k-s]]_q$ MDS EAQECCs for $1\leq s \leq k$ if $q$, $n$ and $k$ satisfy one of the following conditions.
\begin{enumerate}[(1)]
\item $n=q+1$ and $k=\frac{q+1}{2}$.
\item $n>2$ with $(n-1)|(q-1)$ and $-(n-1)$ is a square of $\Ff_q$.
\item $n>1$ with $n|q$.
\item $q\equiv1\pmod {4}$ and $n=2m<q$, where $m>1$ with $m|q$.
\end{enumerate}
\end{thm}

\begin{thm}\label{thm43}
Let $q>3$ be an odd prime power. Then there exist $[[n,k-s,n-k+1;n-k-s]]_q$ and $[[n,n-k-s,k+1;k-s]]_q$ MDS EAQECCs for $1\leq s \leq k-1$ if $q$ and $n$ satisfy one of the following conditions.
\begin{enumerate}[(1)]
\item $n=q+1$.
\item $n>1$ with $n|(q-1)$.
\item $q\equiv1\pmod {4}$ and $n=2m<q-1$, where $m>1$ with $m|(q-1)$.
\end{enumerate}
\end{thm}

\begin{remark}
The required number of maximally entangled states of the MDS EAQECCs reported in the literature (see for instance \cite{EAQ,CC1,CC2,CC3}) is fixed. However, the required number of maximally entangled states of the MDS EAQECCs defined by Theorem \ref{thm41}, \ref{thm44}, \ref{thm42} and \ref{thm43} can take almost all values. Consequently, the parameters of the MDS EAQECCs defined by Theorem \ref{thm41}, \ref{thm44}, \ref{thm42} and \ref{thm43} are new and flexible. For instance, we list the parameters of these MDS EAQECCs for some given $q,n$ in Table \ref{table1}, \ref{table2}, \ref{table3} and \ref{table4}.
\end{remark}

\begin{remark}
To the best of our knowledge, it is the first infinite family of MDS EAQECCs that the required number of maximally entangled states can take all values since ${\rm dim}({\rm Hull}(\cc))$ is arbitrary. (we exclude the case that ${\rm dim}({\rm Hull}(\cc))=0$, i.e., $c=k$ or $n-k$ since the existence question about MDS codes with $0$-dimensional hull has been completely solved for $q=2$ and $q>3$ in \cite{Carlet} and \cite{Jin}.)
\end{remark}

\begin{table}
\caption{\footnotesize Sample parameters of MDS EAQECCs of Theorem \ref{thm41} \protect \\
for $q=16$ and $n=10$}
 \label{table1}
 \centering
\begin{threeparttable}
\begin{tabular*}{8cm}{lll|lll}
\hline
$k$ & $s$  & $[[n,k_1,d_1;c_1]]_q$\tnote{1} & $k$ & $s$  & $[[n,k_2,d_2;c_2]]_q$\tnote{2}\\
\hline
$2$ & $1$ & $[[10, 1, 9,7]]_{16}$ & $2$ & $1$ & $[[10, 7, 3,1]]_{16}$\\
$2$ & $2$ & $[[10, 0, 9,6]]_{16}$ & $2$ & $2$ & $[[10, 6, 3,0]]_{16}$\\
$3$ & $1$ & $[[10, 2, 8,6]]_{16}$ & $3$ & $1$ & $[[10, 6, 4,2]]_{16}$\\
$3$ & $2$ & $[[10, 1, 8,5]]_{16}$ & $3$ & $2$ & $[[10, 5, 4,1]]_{16}$\\
$3$ & $3$ & $[[10, 0, 8,4]]_{16}$ & $3$ & $3$ & $[[10, 4, 4,0]]_{16}$\\
$4$ & $1$ & $[[10, 3, 7,5]]_{16}$ & $4$ & $1$ & $[[10, 5, 5,3]]_{16}$\\
$4$ & $2$ & $[[10, 2, 7,4]]_{16}$ & $4$ & $2$ & $[[10, 4, 5,2]]_{16}$\\
$4$ & $3$ & $[[10, 1, 7,3]]_{16}$ & $4$ & $3$ & $[[10, 3, 5,1]]_{16}$\\
$4$ & $4$ & $[[10, 0, 7,2]]_{16}$ & $4$ & $4$ & $[[10, 2, 5,0]]_{16}$\\
$5$ & $1$ & $[[10, 4, 6,4]]_{16}$ & $5$ & $1$ & $[[10, 4, 6,4]]_{16}$\\
$5$ & $2$ & $[[10, 3, 6,3]]_{16}$ & $5$ & $2$ & $[[10, 3, 6,3]]_{16}$\\
$5$ & $3$ & $[[10, 2, 6,2]]_{16}$ & $5$ & $3$ & $[[10, 2, 6,2]]_{16}$\\
$5$ & $4$ & $[[10, 1, 6,1]]_{16}$ & $5$ & $4$ & $[[10, 1, 6,1]]_{16}$\\
$5$ & $5$ & $[[10, 0, 6,0]]_{16}$ & $5$ & $5$ & $[[10, 0, 6,0]]_{16}$\\
\hline
\end{tabular*}
      \begin{tablenotes}
        \footnotesize
        \item[1] $k_1=k-s$, $d_1=n-k+1$, $c_1=n-k-s$.
        \item[2] $k_2=n-k-s$, $d_2=k+1$, $c_2=k-s$.
      \end{tablenotes}
    \end{threeparttable}
\end{table}

\begin{table}
\caption{\footnotesize Sample parameters of MDS EAQECCs of Theorem \ref{thm44} \protect \\
for $q=16$ and $n=17$}
 \label{table2}
 \centering
\begin{threeparttable}
\begin{tabular*}{8cm}{lll|lll}
\hline
$k$ & $s$  & $[[n,k_1,d_1;c_1]]_q$\tnote{1} & $k$ & $s$  & $[[n,k_2,d_2;c_2]]_q$\tnote{2}\\
\hline
$2$ & $1$ & $[[17, 1, 16,14]]_{16}$ & $2$ & $1$ & $[[17, 14, 3,1]]_{16}$\\
$3$ & $1$ & $[[17, 2, 15,13]]_{16}$ & $3$ & $1$ & $[[17, 13, 4,2]]_{16}$\\
$3$ & $2$ & $[[17, 1, 15,12]]_{16}$ & $3$ & $2$ & $[[17, 12, 4,1]]_{16}$\\
$4$ & $1$ & $[[17, 3, 14,12]]_{16}$ & $4$ & $1$ & $[[17, 12, 5,3]]_{16}$\\
$4$ & $2$ & $[[17, 2, 14,11]]_{16}$ & $4$ & $2$ & $[[17, 11, 5,2]]_{16}$\\
$4$ & $3$ & $[[17, 1, 14,10]]_{16}$ & $4$ & $3$ & $[[17, 10, 5,1]]_{16}$\\
$5$ & $1$ & $[[17, 4, 13,11]]_{16}$ & $5$ & $1$ & $[[17, 11, 6,4]]_{16}$\\
$5$ & $2$ & $[[17, 3, 13,10]]_{16}$ & $5$ & $2$ & $[[17, 10, 6,3]]_{16}$\\
$5$ & $3$ & $[[17, 2, 13,9]]_{16}$ & $5$ & $3$ & $[[17, 9, 6,2]]_{16}$\\
$5$ & $4$ & $[[17, 1, 13,8]]_{16}$ & $5$ & $4$ & $[[17, 8, 6,1]]_{16}$\\
$6$ & $1$ & $[[17, 5, 12,10]]_{16}$ & $6$ & $1$ & $[[17, 10, 7,5]]_{16}$\\
$6$ & $2$ & $[[17, 4, 12,9]]_{16}$ & $6$ & $2$ & $[[17, 9, 7,4]]_{16}$\\
$6$ & $3$ & $[[17, 3, 12,8]]_{16}$ & $6$ & $3$ & $[[17, 8, 7,3]]_{16}$\\
$6$ & $4$ & $[[17, 2, 12,7]]_{16}$ & $6$ & $4$ & $[[17, 7, 7,2]]_{16}$\\
$6$ & $5$ & $[[17, 1, 12,6]]_{16}$ & $6$ & $5$ & $[[17, 6, 7,1]]_{16}$\\
$7$ & $1$ & $[[17, 6, 11,9]]_{16}$ & $7$ & $1$ & $[[17, 9, 8,6]]_{16}$\\
$7$ & $2$ & $[[17, 5, 11,8]]_{16}$ & $7$ & $2$ & $[[17, 8, 8,5]]_{16}$\\
$7$ & $3$ & $[[17, 4, 11,7]]_{16}$ & $7$ & $3$ & $[[17, 7, 8,4]]_{16}$\\
$7$ & $4$ & $[[17, 3, 11,6]]_{16}$ & $7$ & $4$ & $[[17, 6, 8,3]]_{16}$\\
$7$ & $5$ & $[[17, 2, 11,5]]_{16}$ & $7$ & $5$ & $[[17, 5, 8,2]]_{16}$\\
$7$ & $6$ & $[[17, 1, 11,4]]_{16}$ & $7$ & $6$ & $[[17, 4, 8,1]]_{16}$\\
$8$ & $1$ & $[[17, 7, 10,8]]_{16}$ & $8$ & $1$ & $[[17, 8, 9,7]]_{16}$\\
$8$ & $2$ & $[[17, 6, 10,7]]_{16}$ & $8$ & $2$ & $[[17, 7, 9,6]]_{16}$\\
$8$ & $3$ & $[[17, 5, 10,6]]_{16}$ & $8$ & $3$ & $[[17, 6, 9,5]]_{16}$\\
$8$ & $4$ & $[[17, 4, 10,5]]_{16}$ & $8$ & $4$ & $[[17, 5, 9,4]]_{16}$\\
$8$ & $5$ & $[[17, 3, 10,4]]_{16}$ & $8$ & $5$ & $[[17, 4, 9,3]]_{16}$\\
$8$ & $6$ & $[[17, 2, 10,3]]_{16}$ & $8$ & $6$ & $[[17, 3, 9,2]]_{16}$\\
$8$ & $7$ & $[[17, 1, 10,2]]_{16}$ & $8$ & $7$ & $[[17, 2, 9,1]]_{16}$\\

\hline
\end{tabular*}
      \begin{tablenotes}
        \footnotesize
        \item[1] $k_1=k-s$, $d_1=n-k+1$, $c_1=n-k-s$.
        \item[2] $k_2=n-k-s$, $d_2=k+1$, $c_2=k-s$.
      \end{tablenotes}
    \end{threeparttable}
\end{table}

\begin{table}
\caption{\footnotesize Sample parameters of MDS EAQECCs of Theorem \ref{thm42} (3) \protect \\
for $q=81$ and $n=9$}
 \label{table3}
 \centering
\begin{threeparttable}
\begin{tabular*}{8cm}{lll|lll}
\hline
$k$ & $s$  & $[[n,k_1,d_1;c_1]]_q$\tnote{1} & $k$ & $s$  & $[[n,k_2,d_2;c_2]]_q$\tnote{2}\\
\hline
$2$ & $1$ & $[[9, 1, 8,6]]_{81}$ & $2$ & $1$ & $[[9, 6, 3,1]]_{81}$\\
$2$ & $2$ & $[[9, 0, 8,5]]_{81}$ & $2$ & $2$ & $[[9, 5, 3,0]]_{81}$\\
$3$ & $1$ & $[[9, 2, 7,5]]_{81}$ & $3$ & $1$ & $[[9, 5, 4,2]]_{81}$\\
$3$ & $2$ & $[[9, 1, 7,4]]_{81}$ & $3$ & $2$ & $[[9, 4, 4,1]]_{81}$\\
$3$ & $3$ & $[[9, 0, 7,3]]_{81}$ & $3$ & $3$ & $[[9, 3, 4,0]]_{81}$\\
$4$ & $1$ & $[[9, 3, 6,4]]_{81}$ & $4$ & $1$ & $[[9, 4, 5,3]]_{81}$\\
$4$ & $2$ & $[[9, 2, 6,3]]_{81}$ & $4$ & $2$ & $[[9, 3, 5,2]]_{81}$\\
$4$ & $3$ & $[[9, 1, 6,2]]_{81}$ & $4$ & $3$ & $[[9, 2, 5,1]]_{81}$\\
$4$ & $4$ & $[[9, 0, 6,1]]_{81}$ & $4$ & $4$ & $[[9, 1, 5,0]]_{81}$\\
\hline
\end{tabular*}
      \begin{tablenotes}
        \footnotesize
        \item[1] $k_1=k-s$, $d_1=n-k+1$, $c_1=n-k-s$.
        \item[2] $k_2=n-k-s$, $d_2=k+1$, $c_2=k-s$.
      \end{tablenotes}
    \end{threeparttable}
\end{table}

\begin{table}
\caption{\footnotesize Sample parameters of MDS EAQECCs of Theorem \ref{thm43} (2) \protect \\
for $q=27$ and $n=13$}
 \label{table4}
 \centering
\begin{threeparttable}
\begin{tabular*}{8cm}{lll|lll}
\hline
$k$ & $s$  & $[[n,k_1,d_1;c_1]]_q$\tnote{1} & $k$ & $s$  & $[[n,k_2,d_2;c_2]]_q$\tnote{2}\\
\hline
$2$ & $1$ & $[[13, 1, 12,10]]_{27}$ & $2$ & $1$ & $[[13, 10, 3,1]]_{27}$\\
$3$ & $1$ & $[[13, 2, 11,9]]_{27}$ & $3$ & $1$ & $[[13, 9, 4,2]]_{27}$\\
$3$ & $2$ & $[[13, 1, 11,8]]_{27}$ & $3$ & $2$ & $[[13, 8, 4,1]]_{27}$\\
$4$ & $1$ & $[[13, 3, 10,8]]_{27}$ & $4$ & $1$ & $[[13, 8, 5,3]]_{27}$\\
$4$ & $2$ & $[[13, 2, 10,7]]_{27}$ & $4$ & $2$ & $[[13, 7, 5,2]]_{27}$\\
$4$ & $3$ & $[[13, 1, 10,6]]_{27}$ & $4$ & $3$ & $[[13, 6, 5,1]]_{27}$\\
$5$ & $1$ & $[[13, 4, 9,7]]_{27}$ & $5$ & $1$ & $[[13, 7, 6,4]]_{27}$\\
$5$ & $2$ & $[[13, 3, 9,6]]_{27}$ & $5$ & $2$ & $[[13, 6, 6,3]]_{27}$\\
$5$ & $3$ & $[[13, 2, 9,5]]_{27}$ & $5$ & $3$ & $[[13, 5, 6,2]]_{27}$\\
$5$ & $4$ & $[[13, 1, 9,4]]_{27}$ & $5$ & $4$ & $[[13, 4, 6,1]]_{27}$\\
$6$ & $1$ & $[[13, 5, 8,6]]_{27}$ & $6$ & $1$ & $[[13, 6, 7,5]]_{27}$\\
$6$ & $2$ & $[[13, 4, 8,5]]_{27}$ & $6$ & $2$ & $[[13, 5, 7,4]]_{27}$\\
$6$ & $3$ & $[[13, 3, 8,4]]_{27}$ & $6$ & $3$ & $[[13, 4, 7,3]]_{27}$\\
$6$ & $4$ & $[[13, 2, 8,3]]_{27}$ & $6$ & $4$ & $[[13, 3, 7,2]]_{27}$\\
$6$ & $5$ & $[[13, 1, 8,2]]_{27}$ & $6$ & $5$ & $[[13, 2, 7,1]]_{27}$\\
\hline
\end{tabular*}
      \begin{tablenotes}
        \footnotesize
        \item[1] $k_1=k-s$, $d_1=n-k+1$, $c_1=n-k-s$.
        \item[2] $k_2=n-k-s$, $d_2=k+1$, $c_2=k-s$.
      \end{tablenotes}
    \end{threeparttable}
\end{table}

\section{Concluding remarks}
In this paper, we proposed several infinite families of MDS codes with arbitrary dimensional hull by using GRS codes and extended GRS codes. With the MDS codes constructed in Section 3, we presented several infinite families of MDS EAQECCs whose required number of maximally entangled states is flexible. The study of hull of linear codes is an interesting problem in coding theory. We believe that there are some other potential applications in coding theory.


\begin{thebibliography}{1}

\bibitem{Hull}Jr. E.F. Assmus and J. Key, Designs and Their Codes. Cambridge University Press, Cambridge (1992).
Cambridge Tracts in Mathematics, vol.103 (Second printing with corrections, 1993).

\bibitem{Brun}T. Brun, I. Devetak and M.H. Hsieh, ``Catalytic quantum error correction,'' IEEE Trans. Inf. Theory, vol. 60, pp. 3073–
3089, 2014.

\bibitem{C1} T. Brun, I. Devetak and M.H. Hsieh, ``Correcting quantum errors with entanglement,'' Science, vol. 314, pp. 436–439, 2006.

\bibitem{Carlet}C. Carlet, S. Mesnager, C. Tang, Y. Qi, and R. Pellikaan, ``Linear codes
over $\Ff_q$ are equivalent to LCD codes for $q > 3$,'' IEEE Trans. Inform.
Theory, Doi: 10.1109/TIT.2017.2748955, 2017.

\bibitem{Chenbc} B. Chen and H. Liu, ``New constructions of MDS codes with
complementary duals,'' IEEE Trans. Inf. Theory, Doi: 10.1109/TIT.2017.2748955, 2017.

\bibitem{CSS1} A. Calderbank and P. Shor, ``Good quantum error-correcting codes exist,'' Phys. Rev. A, vol. 54, pp. 1098–1105, 1996.

\bibitem{CC2} J. Chen, Y. Huang, C. Feng and R. Chen, ``Entanglement-assisted quantum MDS codes
constructed from negacyclic codes,'' Quantum Inf. Process. 16: 303, https://doi.org/10.1007/s11128-017-1750-4, 2017.

\bibitem{QIP} I. Djordjevic, Quantum Information Processing and Quantum Error
Correction: An Engineering Approach. Waltham, MA, USA: Academic
press, 2012.

\bibitem{EAQ1} M.H. Hsich, I. Devetak, T. Brun, ``General entanglement-assisted quantum error-correcting codes,'' Phys.
Rev. A, 76, 062313, 2007.

\bibitem{EAQ}K. Guenda, S. Jitman, T A. Gulliver, ``Constructions of good entanglementassisted
quanutm error cottecting codes,'' Des. Codes Cryptogr., vol. 86, pp. 121-136, 2018.

\bibitem{Jin}L. Jin, ``Construction of MDS codes with complementary duals,'' IEEE
Trans. Inf. Theory, vol. 63, no. 5, pp. 2843-2847, 2017.

\bibitem{Mac}F. J. MacWilliams and N. J. A. Sloane, The Theory of Error-Correcting
Codes, North-Holland, Amsterdam, 1977.

\bibitem{CC3}L. Lu, R. Li, L. Guo, Y. Ma and Y. Liu, ``Entanglement-assisted quantum MDS codes from
negacyclic codes,'' Quantum Inf. Process, 17: 69, https://doi.org/10.1007/s11128-018-1838-5, 2018.

\bibitem{A1}J.S. Leon, ``An algorithm for computing the automorphism group of a Hadamard matrix,'' J. Comb. Theory, Ser. A, vol. 27, no. 3, pp. 289–306, 1979.

\bibitem{A3}J.S. Leon, ``Computing automorphism groups of error-correcting codes,'' IEEE Trans. Inf. Theory, vol. 28, no. 3, pp. 496–511, 1982.

\bibitem{A4}J.S. Leon, ``Permutation group algorithms based on partition,'' I: theory and algorithms, J. Symb. Comput, vol. 12, pp. 533–583, 1991.

\bibitem{A2}E. Petrank and R.M. Roth, ``Is code equivalence easy to decide?,'' IEEE Trans. Inf. Theory, vol. 43, no. 5, pp. 1602–1604, 1997.

\bibitem{CC1} J. Qian and L. Zhang, ``On MDS linear complementary dual codes and entanglement-assisted quantum codes,'' Des. Codes Cryptogr., vol. 86, no. 7, pp. 1565-1572, 2018.


\bibitem{R1}N. Sendrier, ``On the dimension of the hull,'' SIAM J. Appl. Math., pp. 282–293, 1997.

\bibitem{A5}N. Sendrier, ``Finding the permutationbetween equivalent binary code,'' in: Proceedings of IEEE ISIT’1997, Ulm, Germany, 1997, pp. 367.

\bibitem{A6}N. Sendrier, ``Finding the permutation between equivalent codes: the support splitting algorithm,'' IEEE Trans. Inf. Theory, vo. 46, no. 4, pp. 1193–1203, 2000.

\bibitem{R2} G. Skersys, ``The average dimension of the hull of cyclic codes,'' Discrete Appl. Math., vol. 128, no. 1, pp. 275–292, 2003.

\bibitem{R3} E. Sangwisuta, S. Jitmanb, S. Ling and P. Udomkavanicha, ``Hulls of cyclic and negacyclic codes over finite fields,'' Finite Fields Appl., vol. 33, pp. 232-257, 2015.

\bibitem{CSS2} A. Steane, ``Error-correcting codes in quantum theory,'' Phys. Rev. Lett., vol. 77, pp. 793–797, 1996.

\bibitem{EAQ2} M. Wilde and T. Brun, ``Optimal entanglement formulas for entanglement-assisted quantum coding,''
Phys. Rev. A, 77, 064302, 2008.






\end{thebibliography}
\end{document}